%% file: sep23.tex
\documentclass[english,10pt,letters]{IEEEtran}
\usepackage{amsthm,amssymb,graphicx,graphicx,multirow,amsmath,color}
\usepackage{verbatim}
\usepackage{array}
\usepackage{caption}

\include{notation}

\begin{document}
\title{Coverage Analysis In Downlink Poisson Cellular Network With $\kappa$-$\mu$ Shadowed Fading}

\author{\IEEEauthorblockN{Sudharsan Parthasarathy and  Radha Krishna Ganti}\\
\IEEEauthorblockA{Department of Electrical Engineering\\
Indian Institute of Technology Madras\\
Chennai, India 600036\\
\{sudharsan.p, rganti\}@ee.iitm.ac.in}
}

\maketitle
\begin{abstract}
The downlink coverage probability of a cellular network, when the base station locations are modelled by a Poisson point process (PPP), is known when the desired channel is Nakagami distributed with an integer shape parameter. However, for many interesting fading distributions such as Rician, Rician shadowing, $\kappa$-$\mu$, $\eta$-$\mu$, etc., the coverage probability is unknown. $\kappa$-$\mu$ shadowed fading is a generic fading distribution whose special cases are many of these popular distributions known so far. In this letter, we derive the coverage probability when the desired channel experiences $\kappa$-$\mu$ shadowed fading. Using numerical simulations, we verify our analytical expressions.
\end{abstract}

\section{Introduction}
%
%
%
%
The downlink coverage probability of a single-tier cellular network with distance-dependent interference was analyzed  in \cite{andrews2011tractable}  using tools from stochastic geometry. This was followed by new results for multi-tier cellular networks in single antenna \cite{jo2012heterogeneous} and multi-antenna systems~\cite{dhillon2013downlink}. An important assumption in these works is that the fading distribution of the nearest (desired) base station (BS) is Rayleigh. 
 Rayleigh fading is an important assumption as the channel power will follow exponential distribution. This allows  the distribution of signal-to-interference ratio ($\sir$) to be expressed in terms of the Laplace transform of interference which can easily be computed using standard tools from stochastic geometry. In a multi-antenna system that uses maximal-ratio combining, if the fading distribution of all the links are i.i.d. Rayleigh distributed, the channel power is Gamma distributed with integer shape parameter (equal to the number of antenna terminals). The coverage probability of such a system can be computed by using Laplace transform of interference and its derivative. 

However, for popular  channel fading distribution like Rician, the Laplace trick cannot be used as the complementary cumulative distribution function ($\CCDF$)  of the channel power is not an exponential function. In~\cite{yang2015coverage}, the coverage probability of a two-tier cellular network was obtained when the desired signal experiences Rician fading. This approach assumes that the  $\CCDF$ of a  non-central chi-squared distributed (square of Rician) random variable, can be approximated as a weighted sum of exponentials. The weights and abscissas are obtained by minimizing the mean squared error between the $\CCDF$ and this approximation. As the function minimized is not convex, the weights and abscissas obtained are locally optimal and are highly dependent on the initial points assumed.  Also, there are no closed-form expressions available for these weights and abscissas and have to be computed numerically. In \cite{di2014stochastic}, the coverage probability of a single-tier network with an arbitrary fading distribution was derived using Gil-Pelaez inversion theorem. The coverage probability expression in \cite{di2014stochastic}  requires a numerical integration of the imaginary part of the moment generating function ($\mgf$) of the desired channel's power. However, this approach requires numerical evaluation of multi-dimensional integrals for heterogeneous networks. In \cite{madhu}, it is shown that analytically tractable expressions for coverage probability of a heterogeneous cellular network can not be derived in presence of arbitrary fading if the tier association is based on maximum average received power. In \cite{di2013average}, \cite{alammouri2016modeling} average rate was derived for arbitrary fading channels and fading channels with dominant specular components respectively. But these approaches can not be used for evaluating the coverage probability. 

In this letter, we assume all the channels to experience independent $\kappa$-$\mu$ shadowed fading and derive the exact coverage probability when the parameter $\mu$ of the desired channel is an integer and  use ``Rician approximation" to derive an approximate one when $\mu$ is not an integer.  So using this method, the coverage probability can be obtained if the desired channel is Nakagami faded with non-integer shape parameter whereas the  Laplace trick is useful only when the shape parameter is an integer. As the popular fading distributions such as Rician, Rayleigh, Nakagami, Rician shadowing, $\kappa$-$\mu$, $\eta$-$\mu$ are special cases of $\kappa$-$\mu$ shadowed fading, the  coverage probability expression obtained is generic.  Our analysis assumes that the single tier base stations are PPP distributed and the interfering signals fade independently and identically. The analysis can be easily extended to a multi tier heterogeneous network that uses maximum average received power based association following similar steps as in \cite{jo2012heterogeneous}.

\section{System Model}
\label{sec:systemmodel}

The base stations are modelled by a homogeneous Poisson point process $\Phi $ $\subseteq$ $\R^2$ of intensity $\lambda$. All the base stations are assumed to transmit with unit power. The signal from a base station located at $x \in \mathbb{R}^2,$ experiences a path loss $||x||^{-\alpha}$, where $\alpha > 2.$  Without loss of generality, a typical user is assumed to be at the origin and is associated with the nearest base station located at a distance $r$. Nearest base station association in a single tier network is same as the maximum average received power based association  \cite{jo2012heterogeneous}. This is preferred to the highest $\sir$ based association so that frequent handovers that occurs due to short term fading and shadowing can be avoided \cite{jo2012heterogeneous}.  From \cite{andrews2011tractable}, the nearest neighbour distance is Rayleigh distributed, \ie, $f(r)=2 \pi \lambda r \exp(-\pi \lambda r^2)$. The system is assumed to be interference limited and hence noise is neglected.

\section{$\kappa$-$\mu$ shadowed fading}
\label{sec:kappamu}
$\kappa$-$\mu$ shadowed fading is represented by three parameters viz. $\kappa$, $\mu$ and $m$. Let $\gamma$ denote the signal power. The probability density function of the signal power when the channel experiences $\kappa$-$\mu$ shadowed fading \cite{paris2014statistical} is denoted by $f(\gamma)$ and is given as  $\frac{\mu^\mu m^m (1+\kappa)^\mu \gamma^{\mu-1}}{\Gamma(\mu) \overline{\gamma}^{\mu} (\mu \kappa+m)^m}  e^{-\frac{\mu(1+\kappa) \gamma}{\overline{\gamma}}} {}_1F_1 \left(m;\mu;\frac{\mu^2 \kappa(1+\kappa)}{\mu \kappa+m} \frac{\gamma}{\overline{\gamma}}\right),$ where ${}_1F_1(a;b;z)\overset{\Delta}{=} \sum\limits_{l=0}^\infty \frac{(a)_l}{(b)_l} \frac{z^l}{l!},$  is the confluent hypergeometric function, $(a)_l=\frac{\Gamma(a+l)}{\Gamma(a)}$. The $\pdf$ can be expressed as $f(\gamma) =\sum\limits_{l=0}^{\infty} w_l \frac{e^{-c \gamma} \gamma^{l+\mu-1} c^{l+\mu}}{\Gamma(l+\mu)},$ where $c=\frac{\mu(1+\kappa)}{\overline{\gamma}} $ and
\begin{equation}
w_l=\frac{\Gamma(l+\mu) (m)_l (\frac{\mu \kappa}{ \mu \kappa +m})^l (\frac{m}{\mu \kappa+m})^m}{\Gamma(\mu) l ! (\mu)_l}.
\label{eqn:weights}
\end{equation}
So the $\pdf$ of channel power $f(\gamma)$ can be represented as an infinite  sum of Gamma densities with parameters $(l+\mu,\frac{1}{c})$ and weights $w_l$. The relation between different fading distributions and $\kappa$-$\mu$ shadowed fading are given in \cite{paris2014statistical}, \cite{moreno2015kappa}.  The parameter $m$ in Nakagami-m fading is denoted as $\hat{m}$ to avoid confusion with parameter $m$ in $\kappa$-$\mu$ shadowed fading. Let the channel power of the desired signal $g_0$ be $\kappa$-$\mu$ shadow faded with parameters $\kappa_0$, $\mu_0$, $m_0.$  The interfering signals are independent of each other and the desired signal. All the interfering signals fade identically, but need not be identical to the desired signal. Let the interfering signals be $\kappa$-$\mu$ shadow faded with parameters $\kappa_i$, $\mu_i$, $m_i.$ So the $\pdf$ of the desired channel power $g_0$ is $f(g_0) = \sum\limits_{l=0}^{\infty} w_l \frac{e^{-c_0 g_0} g_0^{l+\mu_0-1} c_0^{l+\mu_0}}{\Gamma(l+\mu_0)}.$ Similarly the  $\pdf$ of  the interference power $g_i$ is $f(g_i) = \sum\limits_{q=0}^{\infty} v_q \frac{e^{-c_i g_i} g_i^{q+\mu_i-1} c_i^{q+\mu_i}}{\Gamma(q+\mu_i)}.$ In the subsequent Section, coverage probability is derived.

\section{Coverage Probability}\label{sec:hetero}
 
The signal to interference ratio of a typical user at distance $r$ from its associated base station is $\mathtt{SIR}$ = $\frac{ g_{0} r^{-\alpha}}{I },$ where $I$=$\sum\limits_{i \in \Phi \setminus B_{0}}   g_{i} |x_{i}|^{-\alpha}$ and $B_{0}$ is the base station that the typical user is associated with. Here, $g_{i}$ is the channel power from the $i$-th base station to the typical user. The coverage probability of a typical user is 
\begin{equation}
P_{c} = \P(\mathtt{SIR}>T)  = \int\limits_{0}^\infty \P(\mathtt{SIR} > T \vert r) 2 \pi \lambda r e^{-\pi \lambda r^2} \d r, \label{eqn:SIRccdf}
\end{equation}

as the distance to the nearest base station is Rayleigh distributed. First we will  derive the exact coverage probability expression when $\mu_0$ is an integer and then derive the approximate one when $\mu_0$ is not an integer. 

\subsection{Integer $\mu_0$}

Rayleigh, Rician, Rician shadowed, Hoyt, $\kappa$-$\mu$, Nakagami (integer shape parameter) are special cases of $\kappa$-$\mu$ shadowed fading where $\mu$ is an integer \cite{paris2014statistical}. In the following theorem we derive the coverage probability when $\mu_0$ is an integer.

\begin{theorem}
If $\mu_0$ is an integer, then coverage probability ($P_c$) is

\begin{equation}
\sum\limits_{l=0}^{\infty}  \sum\limits_{n=0}^{l+\mu_0-1}  \frac{\partial^n}{\partial s^n}  \frac{w_l (-1)^n }{ n! \sum\limits_{q=0}^{\infty} v_q {}_2F_1(q+\mu_i,-\frac{2}{\alpha},1-\frac{2}{\alpha},-\frac{s T c_0}{c_i})} |_{\footnotesize{s=1}}, 
\label{eqn:Pcexact}
\end{equation}
where ${}_2F_1()$ is the Gauss-Hypergeometric function.
\end{theorem}

\begin{proof}
Substituting for $\sir$ in (\ref{eqn:SIRccdf}), coverage probability
\begin{align}
P_{c} 
&=\int\limits_{0}^{\infty} \P(g_{0}>T I r^{\alpha} ) 2 \pi \lambda r e^{-\pi \lambda r^2} \d r.
\label{eqn:Pc}
\end{align}

As $f(g_0)$=$\sum\limits_{l=0}^{\infty} w_l \frac{e^{-c_0 g_0} g_0^{l+\mu_0-1} c_0^{l+\mu_0}}{\Gamma(l+\mu_0)}$ and using $Y$=$c_0 T I  r^{\alpha}$, 
\begin{align}
\P(g_{0}>T I r^{\alpha} ) &= \E_Y \left(\sum\limits_{l=0}^{\infty} w_l \frac{\Gamma(l+\mu_0,Y)}{\Gamma(l+\mu_0)} \right) \label{eqn:forthm3}\\
 &\stackrel{(a)}=  \E_Y \left(\sum\limits_{l=0}^{\infty} w_l \sum\limits_{n=0}^{l+\mu_0-1} e^{-Y}\frac{Y^n}{n!} \right)\\
&=  \sum\limits_{l=0}^{\infty} w_l \sum\limits_{n=0}^{l+\mu_0-1}   \frac{(-1)^n}{n!} \frac{\partial^n}{\partial s^n} L_Y(s)|_{s=1}.
\label{eqn:Pg0}
\end{align}
Since $\mu_0$ and $l$ are integers, $(a)$ follows from the fact that $\frac{\Gamma(q,Y)}{\Gamma(q)}$= $\sum\limits_{n=0}^{q-1} e^{-Y}\frac{Y^n}{n!}$ , for integer $q$.

\begin{equation}
L_Y(s) = \E(e^{-sY}) = \E(e^{-sc_0TIr^{\alpha}})= L_I(s c_0 T r^{\alpha}).
\label{eqn:Ly}
\end{equation}
\begin{align}
L_I(s) &\stackrel{(a)}= \exp \left(-2 \pi \lambda \int\limits_r^{\infty} (1-\E_g(\exp(-s g v^{-\alpha})))v \d v \right) \nonumber \\
&\stackrel{(b)} = e^{ -2 \pi \lambda \sum\limits_{q=0}^{\infty} v_q \int\limits_r^{\infty} (1-\frac{1}{(1+\frac{s v^{-\alpha}}{c_i})^{q+\mu_i}} )v \d v } \nonumber \\
&= e^{-2 \pi \lambda \sum\limits_{q=0}^{\infty} v_q (\frac{r^2}{2} ({}_2F_1(q+\mu_i,-\frac{2}{\alpha},1-\frac{2}{\alpha},-\frac{r^{-\alpha}s}{c_i})-1)) },
\label{eqn:Li}
\end{align}
(a) from \cite{andrews2011tractable} , (b) as the $\pdf$ of interfering signal can be expressed as a weighted sum of Gamma density functions and the weights sum to 1.

Combining (\ref{eqn:Pc}), (\ref{eqn:Pg0}), (\ref{eqn:Ly}), (\ref{eqn:Li}), and by using the fact that the  weights $v_q$ sum to 1, $P_c$ is
\begin{align*}
 \sum\limits_{l=0}^{\infty}  \sum\limits_{n=0}^{l+\mu_0-1} \frac{\partial^n}{\partial s^n} \int\limits_{0}^{\infty} \frac{2 \pi \lambda r  w_l(-1)^n}{n! e^{\pi \lambda r^2 \sum\limits_{q=0}^{\infty} v_q {}_2F_1(q+\mu_i,-\frac{2}{\alpha},1-\frac{2}{\alpha},-\frac{s T c_0}{c_i}) }} \d r |_{s=1}.
\end{align*}

\end{proof}

In practice, only a few weights $w_l$ in (\ref{eqn:weights}) are significant as sum of the weights can be bounded as shown below. From (\ref{eqn:weights}) 
 \begin{align*}
\sum\limits_{l=N+1}^{\infty} w_l &=  \frac{\left(\frac{m}{m+ \kappa \mu} \right)^m }{\Gamma(m)} \sum\limits_{j=0}^{\infty} \frac{\Gamma(m+N+1+j)(\kappa \mu)^{j+N+1}}{\Gamma(2+N+j) (\kappa \mu +  m)^{j+N+1}}  \\ 
&\stackrel{(a)}\approx \frac{\left(\frac{m}{m+ \kappa \mu} \right)^m (N+1)^{m-1} }{\Gamma(m) (\frac{\kappa \mu}{m+\kappa \mu})^{-N-1}} \sum\limits_{j=0}^{\infty} \frac{(\frac{N+1+j+\frac{m}{2}}{N+1})^{m-1}}{ (\frac{\kappa \mu}{m+  \kappa \mu})^{-j}}\\
& \leq   \frac{\left(\frac{m}{m+ \kappa \mu} \right)^m (N+1)^{m-1} }{\Gamma(m) (\frac{m+\kappa \mu}{\kappa \mu})^{N+1}} \sum\limits_{j=0}^{\infty} \frac{(1+j+\frac{m}{2})^{m-1} }{ (\frac{\kappa \mu}{m+  \kappa \mu})^{-j}}\\
& \leq e^{-\mathcal{O}(N)},
\end{align*}
(a) uses Kershaw's approximation, $\frac{\Gamma(k+\frac{\alpha}{2})}{\Gamma(k)}=(k+\frac{\alpha}{4}-\frac{1}{2})^{\frac{\alpha}{2}}$.

The higher order derivatives in \eqref{eqn:Pcexact} is evaluated using Fa{\`a} di Bruno's formula \cite{faa}, \ie, $\frac{\partial^n}{\partial s^n} f(g(s))$=  
$  \sum\limits_{k=1}^n f^{(k)}(g(s)) B_{(n,k)}(g^{(1)}(s), g^{(2)}(s),..,g^{(n-k+1)}(s)) ,$
where $f^{(k)}$, $g^{(k)}$ are the $k$th order derivatives and $B_{n,k}$ is the Bell polynomial. In \eqref{eqn:Pcexact}, $f(g(s))$ is of the form $\frac{1}{g(s)}$. Hence $f^{(k)}(g(s))$ is $(-1)^k k! g(s)^{-k-1}$ and $g^{(k)}(s)$ is
$\sum\limits_{q=0}^{\infty}  \frac{v_q (q+\mu_i)_k (-\frac{2}{\alpha})_k}{(-1)^{-k}(1-\frac{2}{\alpha})_k} {}_2F_1(q+\mu_i+k, -\frac{2}{\alpha}+k,1-\frac{2}{\alpha}+k,-\frac{sTc_0}{c_i})$
where $(.)_k$ is the Pochhammer symbol, \ie, $(a)_k =\frac{\Gamma(a+k)}{\Gamma(a)}$.

\subsection{Non-integer $\mu_0$}
If $\mu_0$ is not an integer, the distribution of $\sir$ can be expressed in terms of fractional derivatives of Laplace transform of interference which leads to intractable expressions. Another approach is to express each of the weighted Gamma $\pdf$ in turn as a weighted sum of Erlang $\pdf$ (Erlang is a special case of Gamma $\pdf$ with integer shape parameters). The parameters of the Erlang density functions and weights can be obtained through a numerical iterative expectation maximization procedure \cite{thummler2006novel}. Alternatively we come up with a technique to approximate the $\pdf$ of Gamma distribution of non integer shape parameters as a weighted sum of Erlang $\pdf$ using a Rician approximation of the Nakagami distribution. The Rician (then called as Nakagami-n) approximation of Nakagami-m distribution was proposed by Nakagami in \cite{nakagami} and has been widely used in wireless communication. The advantage of this method described below is that the weights and Erlang parameters can be pre-computed.

\begin{itemize}
\item{Square root of a Gamma distributed random variable with shape and scale parameters ($l+\mu_0$,$\frac{1}{c_0}$) is Nakagami-m distributed with shape and scale parameters ($l+\mu_0$,$\frac{l+\mu_0}{c_0}$).} \item{Nakagami-m random variable with parameters ($l+\mu_0$, $\frac{l+\mu_0}{c_0}$) can be approximated  by Rician distribution with parameters ($K_l$, $\frac{l+\mu_0}{c_0}$) through moment matching where $l+\mu_0=\frac{(K_l+1)^2}{2 K_l+1},$ $\forall$ $l+\mu_0 \geq 1$ \cite{nakagami}. The original distribution and the approximate distributions are plotted in Fig. \ref{fig:nakagamirician} and it can be observed that the approximation is very tight.} 
\item{ Rician fading is a special case of $\kappa$-$\mu$ shadowed fading with $\mu=1$, $\kappa=K_l$, $m \rightarrow \infty$ \cite{paris2014statistical}. So the $\pdf$ of power of a  Rician faded channel can be expressed as a weighted sum of Erlang $\pdf$ (as $\mu$ is an integer). Hence using this approximate equivalence, Gamma density of non-integer shape parameter can be expressed as a weighted sum of Erlang $\pdf$. }
\end{itemize}
So $f(g_0)$=$\sum\limits_{l=0}^{\infty} w_l \frac{e^{-c_0 g_0} g_0^{l+\mu_0-1} c_0^{l+\mu_0}}{\Gamma(l+\mu_0)}$ can be expressed as $f(g_0) \approx \sum\limits_{l=0}^{\infty} \sum\limits_{p=0}^{\infty} w_l \omega_{pl} \frac{e^{-c_l g_0} g_0^{p} c_l^{p+1}}{\Gamma(p+1)},$ where $\omega_{pl}$=$\frac{e^{-K_l} K_l^p}{p!},$ $c_l=\frac{1+K_l}{\Omega_l}, \Omega_l=\frac{l+\mu_0}{c_0},$ $K_l=l+\mu_0-1+ \sqrt{(l+\mu_0)(l+\mu_0-1)}.$  By following the same steps as in Theorem 1, if $\mu_0$ is not an integer and is greater than 1, then the coverage probability is approximately
\begin{equation}
\sum\limits_{l=0}^{\infty} \sum\limits_{p=0}^{\infty}  \sum\limits_{n=0}^{p}  \frac{\partial^n}{\partial s^n} \frac{w_l  \omega_{pl} (-1)^n}{n! \sum\limits_{q=0}^{\infty} v_q {}_2F_1(q+\mu_i,-\frac{2}{\alpha},1-\frac{2}{\alpha},-\frac{s T c_l}{ c_i})}|_{s=1}.
 \label{eqn:Pcapprox}
\end{equation}

\begin{figure}[ht]
\centering
\includegraphics[height=2.4in,width=3.75in]{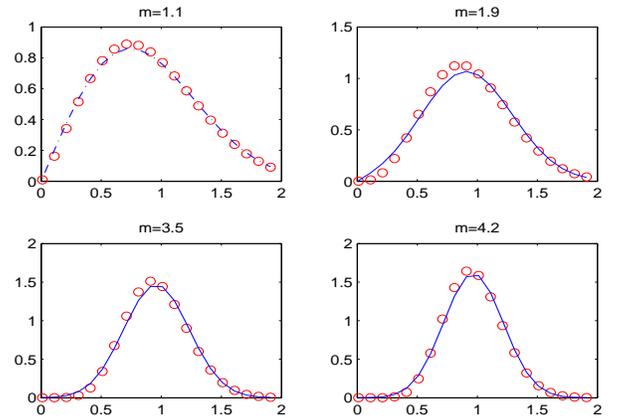}
\caption{Nakagami-$m$ pdf marked by  circles and its Rician approximation.  }
\label{fig:nakagamirician}
\end{figure}%

\section{Numerical Results}
\label{sec:fading}

The results are plotted  for unit mean power in both the desired and interfering channels. We assume identical and independent fading distribution in the desired and interferer links. The coverage probability plots for different fading distributions are provided in Fig. \ref{fig:simulation}. To calculate coverage probability only a finite number of weights $N$ are required  and are provided in  Fig. \ref{fig:simulation}. We observe that simulation results matches closely with the coverage probability derived.  From the plots we can see that in $\kappa$-$\mu$ shadowed fading when $\kappa$ or $\mu$ or $m$ increases, coverage probability increases.   From Fig. \ref{fig:nakagami}, we observe that the Rician approximation of Nakagami (which is used when $\mu$ is a non-integer) is very tight and the accuracy of approximate coverage probability increases with $\sir$ threshold. As the exact coverage probability is not known when $\mu$ is a non-integer, we compute the squared error between the exact and approximate coverage probability when $\mu$ is an integer. As the coverage probability expressions involve  multiple derivatives and summations, deriving an analytical upper bound on the  approximation error is complicated. Hence in Fig. \ref{fig:error}, we plot the squared error for different fading distributions. We observe that as the $\sir$ threshold $T$ increases or with decrease in Nakagami fading or with decrease in $\kappa$, the squared error decreases and is also very low (order of $10^{-5}$).

\begin{figure}[ht]
\centering
\includegraphics[height=2.65in,width=3.75in]{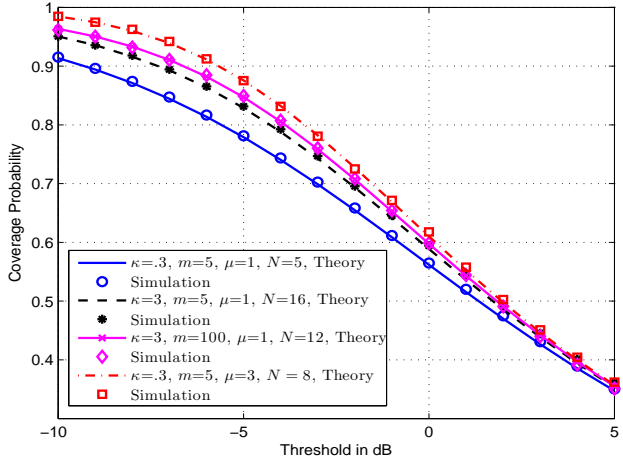}
\caption{Theoretical and simulated coverage probability with $\kappa$-$\mu$ shadowed fading}
\label{fig:simulation}
\end{figure}%

\begin{figure}[ht]
\centering
\includegraphics[height=2.65in,width=3.75in]{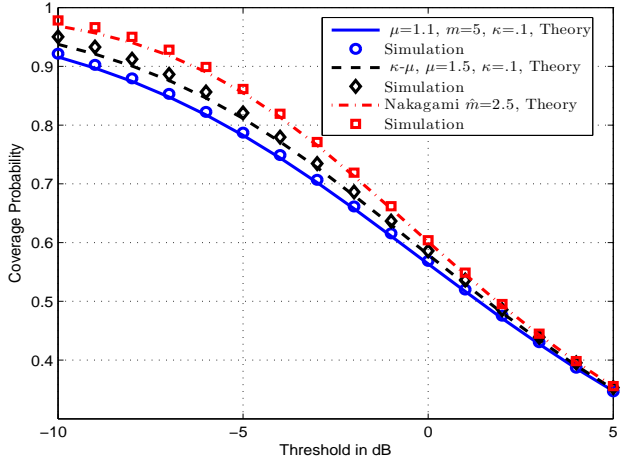}
\caption{Theoretical and simulated coverage probability when $\mu$ is a non-integer. Nakagami-m fading of parameter $\hat{m}$ is a special case of $\kappa$-$\mu$ shadowed fading for $\mu$=$\hat{m}$, $\kappa \rightarrow 0$, $m \rightarrow \infty$ and $\kappa$-$\mu$ fading is a special case of $\kappa$-$\mu$ shadowed fading when $m \rightarrow \infty$   }
\label{fig:nakagami}
\end{figure}%

\begin{figure}[ht]
\centering
\includegraphics[height=2.65in,width=3.75in]{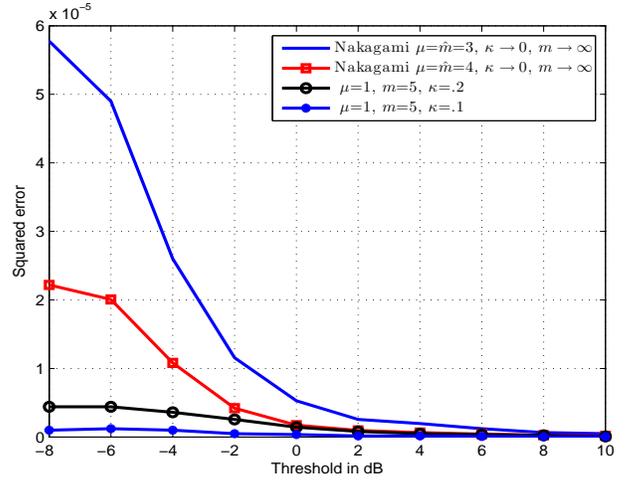}
\caption{Squared error between exact and approximate coverage probability}
\label{fig:error}
\end{figure}%

\section{Conclusion}
In this paper, we have derived the coverage probability when both the desired and interfering links experience $\kappa$-$\mu$ shadowed fading. As $\kappa$-$\mu$ shadowed fading generalizes many popular fading distributions, the coverage probability expression derived can be used when the links experience Rician fading, Nakagami fading, Rician shadowing etc. which were hitherto unknown. By using a Rician approximation, we also derive an approximate coverage probability expression when parameter $\mu$ is not an integer. This is useful in deriving the coverage probability when the shape parameter of Nakagami fading is not an  integer.

\end{document}

%% file: notation.tex
\def\E{\mathbb{E}}

\def\P{\mathbb{P}}

\def\ie{{\em i.e.}}

\def\R{\mathbb{R}}

\def\d{\mathrm{d}}

\def\sir{\mathtt{SIR}}

\def\pdf{\mathtt{PDF}}

\def\CCDF{\mathtt{CCDF}}
\def\mgf{\mathtt{MGF}}

{}
  
  \newtheorem{theorem}{Theorem}